\documentclass{article}
\usepackage[margin=1in]{geometry}
\usepackage{times}
\usepackage[round,authoryear]{natbib}
\usepackage[dvipsnames]{xcolor}
\usepackage{amsmath,amssymb,amsthm,mathtools}
\usepackage{booktabs}
\usepackage{graphicx}
\usepackage{microtype}
\usepackage{enumitem}
\usepackage{array}
\usepackage{multirow}
\usepackage{url}
\usepackage[colorlinks=true,allcolors=blue!55!black]{hyperref}
\usepackage[nameinlink,capitalise,noabbrev]{cleveref}
\usepackage{xspace}

\newtheorem{theorem}{Theorem}
\newtheorem{lemma}[theorem]{Lemma}
\newtheorem{corollary}[theorem]{Corollary}

\theoremstyle{definition}

\newtheorem{remark}[theorem]{Remark}

\newcommand{\relu}{\operatorname{ReLU}}

\newcommand{\TV}{d_{\mathrm{TV}}}
\newcommand{\EE}{\mathbb{E}}
\newcommand{\PP}{\mathbb{P}}
\newcommand{\RR}{\mathbb{R}}
\newcommand{\wt}{\widetilde}
\newcommand{\EXACT}{\textsc{Exact-ReLU-Verify}\xspace}
\newcommand{\one}{\mathbf{1}}
\newcommand{\poly}{\operatorname{poly}}
\newcommand{\unsat}{\operatorname{unsat}}

\hypersetup{
  pdftitle={Random Parameter Noise Does Not Make Exact ReLU Verification Easy},
  pdfauthor={Mojtaba Soltanalian},
  pdfsubject={Smoothed hardness of exact ReLU verification}
}

\title{Random Parameter Noise Does Not Make \\ Exact ReLU Verification Easy}
\author{Mojtaba Soltanalian\\
\small Department of Electrical and Computer Engineering\\
\small University of Illinois Chicago\\
\small \texttt{msol@uic.edu}}

\date{July 2026}

\begin{document}
\maketitle

\begin{abstract}
This technical report studies exact verification of ReLU networks in an adversarial smoothed model.  Every network weight and bias is independently perturbed by Gaussian noise, clipped to $[-2,2]$, and rounded to the exact dyadic grid determined by the input bit complexity.  We show that, under the standard assumption $\mathrm{NP}\not\subseteq\mathrm{BPP}$, there is no sound and complete verifier whose expected running time is polynomial in network size, bit complexity, and inverse noise level for every base instance.  The conclusion already holds at the fixed noise level $\sigma_\star=2^{-11}$ for one-hidden-layer networks over a unit box, with hidden fan-in at most three and base coefficients in $[-1,1]$.

The proof combines an exact gap embedding with a quantitative robustness argument.  A four-ReLU-per-clause construction gives
\[
\max_{x\in[0,1]^n}g_\Phi(x)=\frac{m-\unsat(\Phi)}{3}
\]
for every E3SAT formula $\Phi$ with $m$ clauses; coordinatewise threshold rounding never decreases the objective.  A weighted parameter-sensitivity inequality and Gaussian concentration then show that a verification gap linear in $m$ survives the aggregate perturbation of all coefficients with probability at least $1-e^{-m/8}$.  The proof includes clipping, exact dyadic rounding, output-layer perturbations, polynomial-bit sampling of the rounded Gaussian law, and the conversion from expected smoothed running time to a BPP algorithm.  Computational checks test the exact identity and illustrate the different scaling of extensive and constant gaps; they are diagnostics rather than evidence for the complexity theorem.  The result concerns worst-case base networks in the stated absolute-noise model, but it shows that parameter nondegeneracy alone does not yield a universal smoothed-polynomial guarantee for exact verification.
\end{abstract}

\section{Introduction}

Complete neural-network verification requires a decisive answer: either produce an input violating a specification or prove that no such input exists.  For piecewise-linear networks, this task underlies SMT, mixed-integer, and branch-and-bound verifiers \citep{katz2017reluplex,ehlers2017formal,tjeng2019mip,bunel2020branch,wang2021betacrown}.  Verification competitions document continuing algorithmic progress together with substantial variation in difficulty across benchmarks \citep{kaulen2025vnncomp}.

Worst-case hardness does not by itself determine what happens after perturbation.  One possible explanation for difficult instances is that they depend on exact algebraic alignments, vanishing activation margins, or synchronized coefficients.  Under that explanation, generic parameter noise could remove the relevant structure and make exact verification efficient on every smoothed adversarial instance.

Smoothed analysis formalizes this question.  An adversary first chooses an instance; each numerical parameter is then independently perturbed; complexity is averaged only over the perturbation \citep{spielman2004smoothed}.  This framework explains the practical performance of the simplex method and has revealed both positive and negative phenomena in discrete optimization \citep{roeglin2007integer}.  Neural-network verification adds several requirements: the algorithm must establish a universal inequality over a continuum, it must be correct on every rounded rational instance, and all arithmetic must be charged in the ordinary bit model.

We study the perturb-and-round model in which each network coefficient $q$ becomes
\[
\wt q=\operatorname{clip}_{[-2,2]}
\left(2^{-B}\operatorname{round}\bigl(2^B(q+\sigma Z_q)\bigr)\right),
\qquad Z_q\sim\mathcal N(0,1),
\]
independently.  The input box and output threshold are unchanged.  We prove a conditional lower bound against a universal smoothed-polynomial exact verifier under the standard assumption $\mathrm{NP}\not\subseteq\mathrm{BPP}$.

\paragraph{Quantitative step beyond worst-case hardness.}
A worst-case reduction is not automatically robust to parameter noise.  A constant output gap by itself may be overwhelmed by the aggregate motion of linearly many coefficients.  We therefore use a perfect-completeness gap version of E3SAT, obtain a verification margin proportional to the number of clauses, and compare it with the weighted influence of every perturbed parameter.  Concentration then gives exponentially small label-change probability at a noise level independent of network size.  The argument is a quantitative robustification of gap hardness rather than a claim that the particular clause gadget is the only route to such a result.

\paragraph{What is established.}
\begin{enumerate}[leftmargin=1.55em,itemsep=2pt,topsep=3pt]
\item \textbf{An exact continuous extension of E3SAT.}  A formula with $m$ clauses is represented by a one-hidden-layer network with four ReLUs per clause.  Every fractional point is dominated by its coordinatewise threshold rounding, so the ReLU maximization problem has zero integrality gap and an exact closed form in terms of the minimum number of unsatisfied clauses.

\item \textbf{A weighted smoothed-transfer theorem.}  We prove a deterministic bound of the form
\[
\sup_x |g_{\theta'}(x)-g_\theta(x)|
\le \sum_q c_q|\theta'_q-\theta_q|,
\]
where $c_q$ is the worst-case influence of coefficient $q$.  Gaussian concentration for the weighted folded-normal sum yields an explicit label-preservation probability.  The argument handles clipping, exact dyadic rounding, output-layer perturbations, ReLU boundaries, and the full input box.

\item \textbf{A fixed-noise hardness consequence in the bit model.}  With $\sigma_\star=2^{-11}$, the perturbed instance has the same safe/unsafe answer as the source gap formula with probability at least $1-e^{-m/8}$.  We give a polynomial-bit sampler for the rounded Gaussian law and convert the hypothesized expected running-time bound into a BPP algorithm by Markov truncation.

\item \textbf{Computational checks and reproducibility.}  We compare the exact identity with independent MILP optima, globally optimize perturbed three-input networks by activation-arrangement enumeration, cross-check that solver against MILP, and illustrate the different scaling of extensive and one-defect gaps.  These computations audit implementation-sensitive claims and visualize the proof mechanism; they do not establish the complexity result.
\end{enumerate}

The transfer mechanism is intentionally general: other bounded-coefficient gap reductions with an extensive margin and linear weighted sensitivity can support analogous analyses.  The E3SAT construction is used because it gives an exact identity and transparent coefficient accounting.  The theorem does not claim that typical trained networks are usually hard, nor that perturbations cannot help a particular verifier.  It rules out a guarantee that exact verification becomes smoothed-polynomial for \emph{every} adversarial base network under the stated coefficient perturbation.

\section{Model and main result}\label{sec:model}

A feedforward ReLU network is a finite directed acyclic graph.  Each hidden node computes
\[
z_v=b_v+\sum_{u\to v}w_{vu}y_u,
\qquad y_v=\relu(z_v)=\max\{0,z_v\},
\]
and the scalar output $g_N(x)$ is affine in the final hidden layer.  All weights, biases, box endpoints, and thresholds are rational.

For a rational box $X$, define
\[
\EXACT(N,X,\tau)=\one\{g_N(x)\le \tau\text{ for every }x\in X\}.
\]
An unsafe answer therefore needs a witness $x$ with $g_N(x)>\tau$; a safe answer establishes the universal weak inequality.  The instance size $s$ is the number of neurons plus the number of nonzero affine coefficients.  Its bit complexity $B$ is the largest binary encoding length of any rational datum.  Our construction also works under the narrower convention in which $B$ is computed only from network coefficients.

Let $h=2^{-B}$.  For every explicitly present network weight or bias $q$, independently draw $Z_q\sim\mathcal N(0,1)$ and set
\begin{equation}
\wt q=\Pi_{[-2,2]}
\left(h\operatorname{round}\left(\frac{q+\sigma Z_q}{h}\right)\right),
\label{eq:smoothing}
\end{equation}
where $\Pi_{[-2,2]}$ is clipping and $h\le\sigma\le1$.  Zero hidden biases and the affine output bias are included as parameter slots.  If a convention perturbs fewer slots, every robustness bound below only improves.

\begin{theorem}[Fixed-noise smoothed hardness]\label{thm:main}
Let
\[
\sigma_\star=2^{-11}.
\]
Suppose there is a sound and complete exact algorithm $A$ whose expected bit-running time on \eqref{eq:smoothing} is polynomial in $s$, $B$, and $1/\sigma_\star$, for every adversarial base instance satisfying $2^{-B}\le\sigma_\star$.  Then
\[
\mathrm{NP}\subseteq\mathrm{BPP}.
\]
The implication already holds for one-hidden-layer networks over $[0,1]^{n+1}$, hidden fan-in at most three, base coefficients in $[-1,1]$, $4m+1$ hidden ReLUs, at most $14m+4$ perturbed parameter slots, $s=\Theta(m)$, and $20\le B=O(\log(m+2))$.
\end{theorem}

The extra hidden ReLU is an inactive precision guard.  The semantic SAT construction uses exactly four ReLUs per clause.

\section{An exact four-ReLU-per-clause E3SAT embedding}\label{sec:identity}

We use the following standard consequence of H\aa stad's optimal E3SAT inapproximability theorem \citep{hastad2001optimal}.  Choosing approximation parameter $\epsilon=1/40$, it is NP-hard to distinguish formulas for which all clauses can be satisfied from formulas for which every assignment satisfies at most a $9/10$ fraction.  Equivalently, in the no case every assignment violates at least
\begin{equation}
\eta m,\qquad \eta=\frac{1}{10},
\label{eq:eta}
\end{equation}
clauses.  Taking 64 disjoint copies is a constant-factor padding that preserves both alternatives and ensures $m\ge64$.

For a literal $\lambda$, define its relaxed value
\[
\ell_\lambda(x)=
\begin{cases}
 x_i,&\lambda=X_i,\\
 1-x_i,&\lambda=\neg X_i.
\end{cases}
\]
For each occurrence $\lambda_j$ in a clause $C$, create a literal unit
\begin{equation}
h_{C,j}(x)=\relu\left(\ell_{\lambda_j}(x)-\frac12\right)
\quad\text{with output weight }\frac23.
\label{eq:literal-unit}
\end{equation}
For $C=(\lambda_1\vee\lambda_2\vee\lambda_3)$, let
\[
L_C(x)=\frac{\ell_{\lambda_1}(x)+\ell_{\lambda_2}(x)+\ell_{\lambda_3}(x)}{3}
\]
and create one overflow unit
\begin{equation}
r_C(x)=\relu\left(L_C(x)-\frac13\right)
\quad\text{with output weight }-1.
\label{eq:overflow-unit}
\end{equation}
The network output is
\begin{equation}
g_\Phi(x)=\frac23\sum_{C\in\Phi}\sum_{j=1}^3 h_{C,j}(x)
-\sum_{C\in\Phi}r_C(x).
\label{eq:network}
\end{equation}
Every hidden fan-in is at most three.  Literal-unit weights and biases are $\pm1$ and $\pm1/2$; overflow weights are integer multiples of $1/3$ in $[-1,1]$, overflow biases belong to $\{-1/3,0,1/3,2/3\}$, and output weights are $2/3$ or $-1$.

Let $u(a)$ be the number of clauses false under Boolean assignment $a$, and let
\[
u^\star(\Phi)=\min_{a\in\{0,1\}^n}u(a).
\]

\begin{theorem}[Exact integrality and threshold rounding]\label{thm:identity}
For every E3SAT formula $\Phi$ with $m$ clauses,
\begin{equation}
\max_{x\in[0,1]^n}g_\Phi(x)=\frac{m-u^\star(\Phi)}{3}.
\label{eq:identity}
\end{equation}
More strongly, if $a_i=\one\{x_i\ge1/2\}$, then
\[
g_\Phi(x)\le g_\Phi(a).
\]
Thus every maximum has a Boolean representative.
\end{theorem}

\begin{proof}
Write $\delta_i=\min\{x_i,1-x_i\}$ and let $d_i$ be the number of literal occurrences of variable $X_i$.  The elementary identity
\begin{equation}
\frac23\relu\left(\ell_\lambda(x)-\frac12\right)
=\frac13\left|x_i-\frac12\right|+\frac{\ell_\lambda(x)}{3}-\frac16
\label{eq:fusion}
\end{equation}
uses $|\ell_\lambda(x)-1/2|=|x_i-1/2|$.  Summing \eqref{eq:fusion} over all $3m$ occurrences and using
$L_C-\relu(L_C-1/3)=\min\{L_C,1/3\}$ gives
\begin{equation}
g_\Phi(x)=
\frac13\sum_i d_i\left|x_i-\frac12\right|-\frac{m}{2}
+\sum_{C\in\Phi}\min\left\{L_C(x),\frac13\right\}.
\label{eq:fused-objective}
\end{equation}
Set $V=\sum_i d_i\delta_i$.  Since $\sum_i d_i=3m$, the first two terms in \eqref{eq:fused-objective} equal $-V/3$.

Let $F(a)$ be the set of clauses false under the rounded assignment $a$, and put $u=|F(a)|$.  For $C\in F(a)$, every relaxed literal value equals the corresponding $\delta_i$.  If $D_C$ is the sum of the three occurrence distances, then
\[
\frac13-\min\left\{\frac{D_C}{3},\frac13\right\}
=\frac13\max\{0,1-D_C\}
\ge\frac{1-D_C}{3}.
\]
The reward deficits of clauses satisfied by $a$ are nonnegative.  Therefore
\begin{align*}
\frac m3-g_\Phi(x)
&\ge \frac V3+\frac u3-\frac13\sum_{C\in F(a)}D_C\\
&\ge \frac u3,
\end{align*}
where the final inequality holds because false-clause occurrences are a subset of all literal occurrences, so $\sum_{C\in F(a)}D_C\le V$.

At the Boolean point $a$, each satisfied clause contributes exactly $1/3$ and each false clause contributes zero; hence $g_\Phi(a)=(m-u)/3$.  The displayed inequality is exactly $g_\Phi(x)\le g_\Phi(a)$.  Maximizing over $x$ and choosing an assignment minimizing $u$ proves \eqref{eq:identity}.
\end{proof}

The exact identity yields an extensive symmetric verification gap.  Set
\begin{equation}
T=\frac m3,
\qquad
\tau=T-\frac{\eta m}{6}=\frac{19m}{60},
\qquad
\Gamma=\frac{\eta m}{6}=\frac m{60}.
\label{eq:threshold-margin}
\end{equation}
If $\Phi$ is satisfiable, some Boolean input has $g_\Phi=T=\tau+\Gamma$.  In the no case, \cref{thm:identity} gives
\[
\max_xg_\Phi(x)\le T-\frac{\eta m}{3}=\tau-\Gamma.
\]
The two alternatives are separated from the unchanged threshold by the same margin $\Gamma$.

\section{Weighted robustness under perturbation}\label{sec:robustness}

\subsection{The precision guard}

The smoothing model requires $2^{-B}\le\sigma$.  To make the fixed $\sigma_\star$ admissible under any natural convention for $B$, add one input $x_0\in[0,1]$ and one hidden unit
\begin{equation}
d(x_0)=\relu\left(-1+2^{-20}x_0\right)
\quad\text{with output weight }1.
\label{eq:guard}
\end{equation}
Its base output is identically zero, so it does not change \eqref{eq:identity} or the threshold gap.  Its coefficient $2^{-20}$ ensures $B\ge20$ and therefore
\[
2^{-B}\le2^{-20}<2^{-11}=\sigma_\star.
\]
The guard's perturbation is included in all bounds below.

\subsection{A deterministic weighted sensitivity inequality}

Let $\theta$ denote the vector of all network parameter slots, including hidden biases, output weights, and the output bias.  For each slot $q$, define a nonnegative influence coefficient $c_q$:
\begin{center}
\begin{tabular}{lc}
\toprule
parameter class & influence $c_q$\\
\midrule
hidden input weight or hidden bias & $2$\\
literal-unit output weight & $1/2$\\
overflow-unit output weight & $2/3$\\
precision-guard output weight & $0$\\
output bias & $1$\\
\bottomrule
\end{tabular}
\end{center}

\begin{lemma}[Uniform weighted sensitivity]\label{lem:sensitivity}
For every clipped coefficient vector $\theta'\in[-2,2]^P$ on the same graph,
\begin{equation}
\sup_{x\in[0,1]^{n+1}}
|g_{\theta'}(x)-g_\theta(x)|
\le \sum_{q=1}^P c_q|\theta'_q-\theta_q|.
\label{eq:weighted-sensitivity}
\end{equation}
The inequality is valid without fixing an activation pattern and remains valid when a base or perturbed preactivation equals zero.
\end{lemma}

\begin{proof}
For hidden unit $v$, define
\[
e_v=|b'_v-b_v|+\sum_j|w'_{vj}-w_{vj}|.
\]
Because every input coordinate lies in $[0,1]$,
$|z'_v(x)-z_v(x)|\le e_v$.  ReLU is 1-Lipschitz, so
$|y'_v(x)-y_v(x)|\le e_v$.

Let $a_v$ and $a'_v$ be the base and perturbed output weights.  Clipping gives $|a'_v|\le2$.  If $Y_v$ is a uniform upper bound on the base activation, then
\begin{align*}
|a'_vy'_v-a_vy_v|
&=|(a'_v-a_v)y_v+a'_v(y'_v-y_v)|\\
&\le Y_v|a'_v-a_v|+2e_v.
\end{align*}
For literal units $Y_v=1/2$; for overflow units $Y_v=2/3$; for the inactive precision guard $Y_v=0$.  Summing the displayed bound and adding the output-bias error gives exactly \eqref{eq:weighted-sensitivity}.
\end{proof}

Define
\begin{equation}
C=\sum_qc_q,
\qquad
D^2=\sum_qc_q^2,
\qquad
\mu=\EE|Z|=\sqrt{2/\pi}.
\label{eq:CD}
\end{equation}
The compact construction admits tight linear bounds.  Per clause, the three literal units contribute $27/2$ to $C$ and $99/4$ to $D^2$; the overflow unit contributes at most $26/3$ and $148/9$.  The precision guard and output bias add at most $5$ to $C$ and $9$ to $D^2$.  Hence, for $m\ge64$,
\begin{equation}
C\le \frac{133m}{6}+5\le23m,
\qquad
D^2\le\frac{1483m}{36}+9\le43m.
\label{eq:CD-bounds}
\end{equation}

\subsection{A reusable Gaussian transfer theorem}

\begin{theorem}[Weighted smoothed-gap transfer]\label{thm:transfer}
Suppose a verification instance has symmetric margin $\Gamma>0$ around its threshold and satisfies \eqref{eq:weighted-sensitivity}.  Let $h=2^{-B}$ and assume every base coefficient lies in $[-1,1]$.  If
\begin{equation}
\Lambda
:=\frac{\Gamma}{\sigma}
-\left(\mu+\frac{h}{2\sigma}\right)C
>0,
\label{eq:lambda}
\end{equation}
then perturbation \eqref{eq:smoothing} changes the exact safe/unsafe answer with probability at most
\begin{equation}
\exp\left(-\frac{\Lambda^2}{2D^2}\right).
\label{eq:transfer-prob}
\end{equation}
\end{theorem}

\begin{proof}
Nearest-grid rounding changes a scalar by at most $h/2$.  Projection onto $[-2,2]$ is nonexpansive relative to every base point $q\in[-1,1]$, so
\begin{equation}
|\wt q-q|\le\sigma|Z_q|+\frac h2.
\label{eq:scalar-drift}
\end{equation}
Combining \eqref{eq:scalar-drift} with \cref{lem:sensitivity} gives
\[
\sup_x|g_{\wt\theta}(x)-g_\theta(x)|
\le \sigma F(Z)+\frac{hC}{2},
\qquad
F(Z)=\sum_qc_q|Z_q|.
\]
The function $F:\RR^P\to\RR$ is $D$-Lipschitz in Euclidean norm, and $\EE F=\mu C$.  Gaussian concentration for Lipschitz functions \citep{boucheron2013concentration} implies
\[
\PP\{F-\mu C\ge t\}
\le\exp\left(-\frac{t^2}{2D^2}\right).
\]
A label change requires the uniform output displacement to reach at least $\Gamma$.  This implies $F-\mu C\ge\Lambda$, and \eqref{eq:transfer-prob} follows.
\end{proof}

\begin{remark}[Scope of the transfer theorem]\label{rem:scope}
The theorem is not specific to the E3SAT embedding.  For example, if $G=(V,E)$ is an unweighted graph, then
\[
g_G(x)=\sum_{\{i,j\}\in E}
\bigl(\relu(x_i-x_j)+\relu(x_j-x_i)\bigr)
=\sum_{\{i,j\}\in E}|x_i-x_j|
\]
has maximum $\operatorname{MaxCut}(G)$ over $[0,1]^{|V|}$: with all other coordinates fixed, the objective is convex in one coordinate, so an endpoint choice cannot decrease it, and repeated endpoint rounding yields a Boolean maximizer.  This two-ReLU-per-edge representation is closely related to existing shallow-network hardness constructions \citep{froese2025complexity}.  Combined with standard constant-factor gap hardness for MAX-CUT \citep{hastad2001optimal} and the same linear coefficient accounting, it gives a parallel route to a qualitative fixed-noise statement.  We retain the E3SAT instantiation because its exact optimum, coefficient bounds, and preservation constants are particularly direct.
\end{remark}

\subsection{Explicit fixed-noise instantiation}

For our reduction, $\Gamma=m/60$, $\sigma_\star=2^{-11}$, and $h\le2^{-20}<\sigma_\star$.  Using the coarser but convenient inequality $h/(2\sigma_\star)\le1/2$, $\mu<5/6$, and \eqref{eq:CD-bounds},
\begin{align}
\Lambda
&\ge \frac{m/60}{2^{-11}}
-\left(\frac56+\frac12\right)23m\\
&=\frac{512m}{15}-\frac{92m}{3}
=\frac{52m}{15}.
\label{eq:lambda-lower}
\end{align}
Consequently,
\begin{equation}
\frac{\Lambda^2}{2D^2}
\ge
\frac{(52m/15)^2}{86m}
=\frac{2704}{19350}m
>\frac m8.
\label{eq:exponent}
\end{equation}

\begin{corollary}[Exponential label preservation]\label{cor:preservation}
For every padded gap instance in the reduction,
\[
\PP\{\text{the perturbation changes the exact answer}\}
\le e^{-m/8}.
\]
In particular, because $m\ge64$, the preservation probability exceeds $1-e^{-8}>0.9996$.
\end{corollary}

This is the central smoothed-hardness mechanism.  Both the logical margin and the weighted aggregate perturbation scale linearly with $m$; a sufficiently small \emph{constant} noise level leaves a positive linear separation, and concentration makes failure exponentially unlikely.

\section{From label preservation to a BPP collapse}\label{sec:bpp}

We now complete the proof of \cref{thm:main}.  The construction is polynomial-time and has $4m+1$ hidden ReLUs.  Counting all possible sparse incoming weights, hidden biases, output weights, the precision guard, and the output bias gives at most
\begin{equation}
P\le14m+4
\label{eq:param-count}
\end{equation}
perturbed slots.  After removing unused variables, $n\le3m$; under any standard convention counting input and output nodes, the instance size remains $s=\Theta(m)$.  All semantic coefficients have constant bit length, the guard gives $B\ge20$, and the threshold $19m/60$ gives $B=O(\log(m+2))$ if thresholds are included in $B$.

Assume the verifier $A$ in \cref{thm:main} exists.  Let
\[
Q(s,B,1/\sigma_\star)
\]
be a polynomial upper bound on its expected running time under the exact perturbation law, including any internal randomness.  The polynomial may be hard-coded into the derived machine because the complexity implication is for the fixed algorithm $A$.  Run $A$ for at most $100Q$ bit operations.  Markov's inequality bounds the timeout probability under the exact smoothing distribution by $1/100$.

A standard probabilistic Turing machine cannot sample an ideal real Gaussian at unit cost.  The dyadic rounding makes the required discrete law efficiently approximable.  For one coefficient, the support is
\[
G_B=\{-2,-2+h,\ldots,2-h,2\},
\qquad |G_B|=2^{B+2}+1=\poly(s).
\]
Each interior atom probability is a difference of normal CDF values,
\begin{equation}
p_q(r)=
\Phi\left(\frac{r+h/2-q}{\sigma_\star}\right)
-
\Phi\left(\frac{r-h/2-q}{\sigma_\star}\right),
\label{eq:atom}
\end{equation}
with the two endpoint atoms absorbing the tails.  The normal CDF is polynomial-time computable to any requested number of bits \citep{ko1991real}.  Approximating each marginal by a dyadic probability table within total variation $\epsilon/P$, and sampling that table exactly, produces a joint distribution within $\epsilon$ of the exact independent product law.  Full normalization and bit-complexity details are in \cref{app:sampling}.

Choose $\epsilon=1/100$.  Total variation cannot increase after applying the verifier, its internal randomness, and the timeout rule.  The resulting polynomial-time randomized SAT algorithm can fail only if:
\begin{enumerate}[leftmargin=1.6em,itemsep=1pt,topsep=2pt]
\item the perturbation changes the reduction label, probability at most $e^{-8}$;
\item the exact verifier exceeds its timeout, probability at most $1/100$; or
\item approximate sampling changes the final event, probability at most $1/100$.
\end{enumerate}
The total is below $1/3$.  Accept exactly when $A$ reports \emph{unsafe}.  A satisfiable source instance maps to an unsafe network, while a no instance maps to a safe network, except on the accounted-for perturbation event.  Thus SAT lies in BPP and $\mathrm{NP}\subseteq\mathrm{BPP}$.

The verifier itself is never allowed to be approximate or wrong: approximation is used only by the outer randomized reduction to sample a distribution close to the exact rounded Gaussian law.

\section{Computational checks and numerical illustration}\label{sec:experiments}

The theorem is analytic, and experiments cannot establish smoothed complexity.  We use computations only to test implementation-sensitive consequences of the proof: the exact continuous-to-Boolean identity, global optimization after perturbation, and the different scaling of extensive and constant gaps.  All raw data, fixed seeds, and scripts are included in the reproducibility package.

\subsection{Exact identity and exhaustively optimized perturbations}

\paragraph{Identity audit.}
We generated 18 random E3SAT formulas with 4--7 variables and 18--47 clauses.  The minimum number $u^\star$ of false clauses was found by exhaustive Boolean enumeration.  We then built the compact networks, containing 72--188 ReLUs, and independently maximized them with exact logical big-$M$ formulations using interval-tight preactivation bounds.  The prediction $(m-u^\star)/3$ agreed with the reported global MILP optimum to maximum absolute error $7.11\times10^{-15}$.

\paragraph{Three-dimensional exhaustive laboratory.}
For perturbation studies we use two eight-clause blocks on three variables.  The \emph{gap block} contains all eight sign patterns and therefore falsifies exactly one clause under every Boolean assignment.  The same-size \emph{satisfiable block} removes the clause false at $000$ and duplicates one clause true at $000$.  Each compact block contains 32 ReLUs and 112 coefficient slots before the shared output bias.

A shallow ReLU network is affine on every cell of its activation-hyperplane arrangement.  In three dimensions, a maximum on a box occurs at an arrangement vertex.  For every perturbed block we enumerated all
\[
\binom{32+6}{3}=8{,}436
\]
triples of activation planes and box facets, discarded singular or infeasible intersections, and evaluated the network at every retained point.  We generated 300 globally optimized blocks of each family at each of ten noise levels from $0.0015$ through $0.063$, with $B=16$.  Forty-four independently selected optima were cross-checked against a separate MILP model; the maximum discrepancy was $7.99\times10^{-15}$.

Disjoint blocks have disjoint inputs and additive outputs, so the product optimum is exactly the sum of the block optima plus one independently perturbed output bias.  We formed 50,000 bootstrap products for each noise/size cell and each
$q\in\{1,2,4,8,16,32,64,128\}$, reaching 4,096 ReLUs and $112q+1=14{,}337$ perturbed slots.

\begin{figure}[t]
\centering
\includegraphics[width=\textwidth]{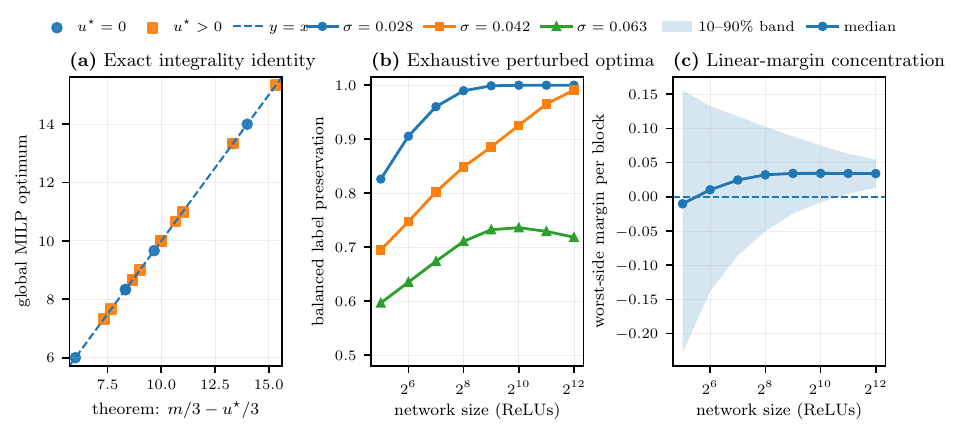}
\caption{\textbf{Exact algebra and exhaustive low-dimensional perturbation study.}
(a) The closed-form optimum in \cref{thm:identity} versus independent global MILP solves.  (b) Balanced preservation of satisfiable and extensively unsatisfiable labels for products of globally optimized three-input blocks.  At $\sigma=0.042$, preservation rises from $0.695$ for one block to $0.991$ for 128 blocks.  At $\sigma=0.063$, systematic side bias prevents convergence to one.  (c) At $\sigma=0.042$, the worst-side margin per block concentrates around a positive constant as the number of blocks grows.  The band is the empirical 10th--90th percentile.  }
\label{fig:exact}
\end{figure}

\Cref{fig:exact}(b)--(c) isolates the extensive-gap phenomenon.  At intermediate noise, one block is unreliable, yet independent products become highly reliable because their signed margin grows and concentrates linearly.  At larger noise, an asymmetric shift in the perturbed optimum can dominate; replication then amplifies the bias.  The theorem avoids relying on cancellation or empirical symmetry by using the deterministic sensitivity certificate.

\subsection{The deterministic certificate exhibits the predicted scaling}

For a realized perturbation, define the exact certificate quantity
\begin{equation}
W(\Delta\theta)=\sum_qc_q|\Delta\theta_q|.
\label{eq:certificate}
\end{equation}
By \cref{lem:sensitivity}, the answer is deterministically preserved whenever $W(\Delta\theta)<\Gamma$.  We compare two block products built from the same parameterized networks:
\begin{itemize}[leftmargin=1.5em,itemsep=1pt,topsep=2pt]
\item in the \emph{extensive-gap} family, all $q$ blocks are gap blocks and $\Gamma=q/6$;
\item in the \emph{one-defect} ablation, only one block carries the logical defect, so $\Gamma=1/6$ while the number of perturbed parameters still grows with $q$.
\end{itemize}

For each $q\in\{1,4,16,64,128\}$ and 37 logarithmically spaced noise levels, we drew 12,000 exact rounded/clipped single-block perturbations at $B=24$ and formed 40,000 bootstrap products.  The extensive family crosses 50\% certificate probability at essentially constant noise, from $1.154\times10^{-3}$ at $q=1$ to $1.132\times10^{-3}$ at $q=128$.  The fitted log--log slope is $-0.0035$.  The one-defect transition falls from $1.154\times10^{-3}$ to $8.49\times10^{-6}$, with slope $-1.014$, matching the predicted $q^{-1}$ law.

\begin{figure}[t]
\centering
\includegraphics[width=\textwidth]{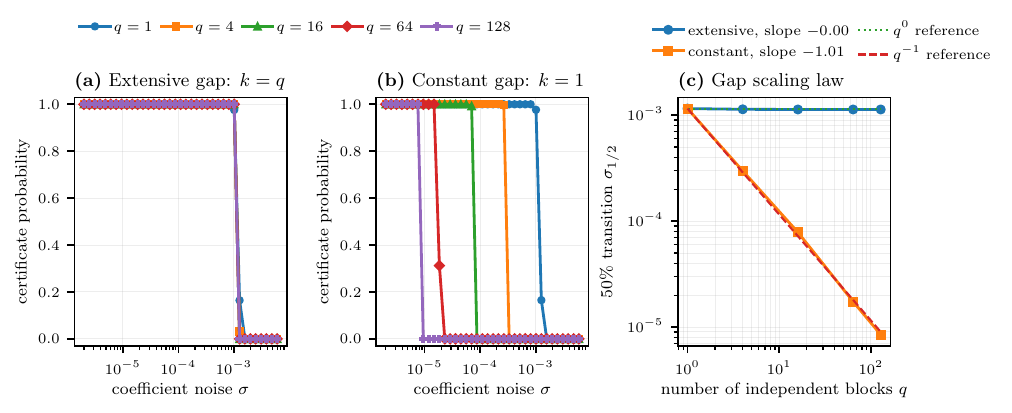}
\caption{\textbf{The extensive verification gap is the fixed-noise mechanism.}
(a) With gap $\Gamma=q/6$, the weighted-certificate transition is essentially independent of the number of perturbed blocks.  (b) With fixed gap $\Gamma=1/6$, admissible noise shrinks as the network grows.  (c) Empirical 50\% transition scales: slope $-0.00$ for the extensive gap and $-1.01$ for the one-defect ablation, compared with $q^0$ and $q^{-1}$ references.}
\label{fig:certificate}
\end{figure}

This ablation illustrates the quantitative condition used by the proof.  In the present construction, fixed noise is possible because the logical margin and aggregate coefficient influence grow at the same extensive scale.

\section{Related work}

\paragraph{Exact ReLU verification.}
Reluplex and related SMT methods \citep{katz2017reluplex,ehlers2017formal}, MILP formulations \citep{tjeng2019mip,anderson2020strong}, and branch-and-bound methods \citep{bunel2020branch,wang2021betacrown} exploit the piecewise-linear structure of ReLU networks while retaining completeness.  Convex relaxations provide scalable bounds but face intrinsic tightness barriers without branching or additional structure \citep{salman2019convex}.  Reachability and verification are NP-hard for highly restricted shallow networks \citep{salzer2021reachability,froese2025complexity}; more recent work gives strong parameterized and approximation hardness for shallow ReLU optimization and verification \citep{froese2025parameterized}.  Worst-case hardness alone does not rule out smoothed polynomial time, because a hard family may be unstable under the prescribed perturbation.

\paragraph{Gap reductions and robustness.}
The present proof should be viewed as a quantitative robustification of gap hardness.  The weighted transfer theorem applies whenever a reduction supplies a sufficiently large verification margin relative to its aggregate coefficient sensitivity.  As noted in \cref{rem:scope}, cut-based shallow-network constructions provide another natural route.  The contribution here is therefore not that the E3SAT gadget is uniquely capable of surviving noise, but that the perturb-and-round model, exact rational correctness, fixed-noise preservation, sampling law, and expected-time conversion are handled together with explicit constants.

\paragraph{Smoothed complexity.}
Smoothed analysis takes a worst case over base inputs and an average over local perturbations \citep{spielman2004smoothed}.  Integer programming demonstrates that smoothing can improve algorithms without erasing all discrete structure \citep{roeglin2007integer}.  In the present setting, the source inapproximability gap is converted into an extensive verification margin and compared directly with weighted coefficient motion.

\paragraph{Parameter smoothing in neural learning.}
\citet{daniely2023smoothness} study computational hardness of learning neural networks under parameter smoothing.  Learning and exact verification differ in their quantifiers, objectives, and correctness requirements.  Our result concerns universal exact verification of a supplied rational network, uses a standard NP-hard gap source, and derives the collapse $\mathrm{NP}\subseteq\mathrm{BPP}$ rather than relying on an average-case conjecture.

\section{Limitations and conclusion}

The theorem concerns worst-case base networks in a specific absolute-parameter-noise model.  It does not predict verifier runtime on networks trained from a particular data distribution, and it does not preclude polynomial algorithms under monotonicity, fixed input dimension, bounded unstable-neuron count, special architectures, or training-induced margins.  The input box and threshold are not perturbed.  The explicit constant $2^{-11}$ is conservative and small relative to the coefficient range, rather than optimized.  Similar conclusions may be obtainable from other gap reductions satisfying the transfer theorem's sensitivity condition.

The computational study uses separable low-dimensional products so that every perturbed block can be globally optimized and independently cross-checked.  It is a diagnostic illustration of the proof mechanism, not a benchmark of modern vision-scale verifiers and not evidence that typical trained networks are difficult after perturbation.

Within these limits, the theorem gives a precise negative result for the stated smoothed model: independent parameter noise and generic activation geometry do not by themselves guarantee smoothed-polynomial exact verification for every adversarial base instance.  The main technical content is the explicit robust-gap accounting, including exact dyadic rounding, clipping, bit-complexity sampling, and expected-time conversion.

\paragraph{Reproducibility.}
The accompanying archive contains the full LaTeX source, fixed random seeds, raw CSV files, figure-generation scripts, proof-audit code, and instructions for rerunning the computations.  The appendices below record the encoding, concentration, sampling, and bit-complexity details.

\section*{Acknowledgments}
The author thanks Prof. Vincent Froese for helpful feedback, in particular for emphasizing the relationship between the smoothed-hardness argument and existing gap reductions for shallow ReLU optimization. Large language-model tools were used substantially for proof exploration, adversarial auditing, experiment design, and manuscript drafting.

\clearpage
\appendix

\section{Full reduction and encoding details}\label{app:reduction}

\subsection{Source promise problem and constant padding}

H\aa stad's E3SAT theorem implies that for every fixed $\epsilon>0$, it is NP-hard to distinguish a satisfiable E3SAT formula from one in which no assignment satisfies more than a $7/8+\epsilon$ fraction of clauses \citep{hastad2001optimal}.  Set $\epsilon=1/40$.  Then $7/8+1/40=9/10$, giving \eqref{eq:eta}.  If the resulting formula has $m_0$ clauses, make 64 disjoint variable-renamed copies.  The new formula has $m=64m_0\ge64$ clauses, remains satisfiable in the yes case, and has the same maximum satisfiable fraction in the no case.  This is a fixed constant blowup.

\subsection{Coefficient expansion}

For a positive literal, \eqref{eq:literal-unit} has incoming weight $1$ and bias $-1/2$; for a negative literal it has weight $-1$ and bias $1/2$.  For a clause with $k$ negative literals,
\[
L_C(x)-\frac13
=\sum_{i\in P_C}\frac{x_i}{3}
-\sum_{i\in N_C}\frac{x_i}{3}
+\frac{k-1}{3}.
\]
If a variable is repeated, equal-variable coefficients are combined and exact zeros are omitted.  The resulting fan-in is at most three, every incoming coefficient lies in $[-1,1]$, and the bias belongs to $\{-1/3,0,1/3,2/3\}$.  The output weights $2/3$ and $-1$ also lie in $[-1,1]$.

The precision guard \eqref{eq:guard} has one weight $2^{-20}$, bias $-1$, and output weight $1$.  Its base preactivation is at most $-1+2^{-20}<0$ on $[0,1]$.

\subsection{Counts}

The semantic construction has $3m$ literal units and $m$ overflow units.  Its parameter slots are bounded by
\begin{center}
\begin{tabular}{lrr}
\toprule
component & hidden units & parameter slots\\
\midrule
literal units & $3m$ & $9m$\\
overflow units & $m$ & at most $5m$\\
precision guard & $1$ & $3$\\
output bias & -- & $1$\\
\midrule
total & $4m+1$ & at most $14m+4$\\
\bottomrule
\end{tabular}
\end{center}
A parameter slot is counted even when its base value is zero.  The instance-size definition counts only nonzero coefficients, so this table is an upper bound for that contribution.  There are at most $3m$ used Boolean variables, and thus the total number of input, hidden, and output nodes plus nonzero coefficients is between constant multiples of $m$.

All semantic coefficients have constant bit length.  The guard gives $B\ge20$.  The threshold $19m/60$ and the formula encoding have $O(\log(m+2))$-bit rational data, so $B=O(\log(m+2))$ under the broad convention that includes the threshold and box endpoints.  Under a network-only convention, $B$ is a fixed constant at least 20.  Both conventions make $2^{-B}\le2^{-20}<\sigma_\star$ and make the rounded support polynomial in the input length.

\subsection{Strict unsafe and weak safe inequalities}

In the satisfiable case the Boolean witness has base output $\tau+\Gamma$.  On the event
\[
\sup_x|g_{\wt\theta}(x)-g_\theta(x)|<\Gamma,
\]
that same witness remains strictly above $\tau$, as required for an unsafe answer.  In the no case every base output is at most $\tau-\Gamma$, so every perturbed output is strictly below $\tau$.  The safe weak inequality therefore holds.  The concentration proof treats every perturbation with uniform displacement at least $\Gamma$, including equality, as part of the failure event.  No correctness claim therefore depends on a tie convention at the output threshold.

\section{Full robustness calculations}\label{app:robustness}

\subsection{Derivation of the influence norms}

For each literal unit there is one incoming weight, one hidden bias, and one output weight.  The incoming parameters have influence 2 and the output parameter has influence $1/2$.  Three literal units per clause therefore contribute
\[
3\left(2+2+\frac12\right)=\frac{27}{2}
\]
to $C$, and
\[
3\left(2^2+2^2+\left(\frac12\right)^2\right)=\frac{99}{4}
\]
to $D^2$.

An overflow unit has at most three incoming weights, one hidden bias, and one output weight.  It contributes at most
\[
4\cdot2+\frac23=\frac{26}{3}
\]
to $C$, and
\[
4\cdot2^2+\left(\frac23\right)^2=\frac{148}{9}
\]
to $D^2$.  Thus the semantic per-clause totals are
\[
\frac{27}{2}+\frac{26}{3}=\frac{133}{6},
\qquad
\frac{99}{4}+\frac{148}{9}=\frac{1483}{36}.
\]
The guard's incoming weight and bias contribute 4 to $C$ and 8 to $D^2$; its output weight has influence zero because the base activation is identically zero.  The output bias contributes 1 to both $C$ and $D^2$.  This gives the first inequalities in \eqref{eq:CD-bounds}; $m\ge64$ yields the simplified bounds.

\subsection{Clipping and rounding}

Let $R_h(t)=h\operatorname{round}(t/h)$.  Every nearest-grid rule satisfies $|R_h(t)-t|\le h/2$.  Since $q\in[-1,1]\subset[-2,2]$ and Euclidean projection onto an interval is nonexpansive relative to a point in that interval,
\[
|\Pi_{[-2,2]}(R_h(q+\sigma Z))-q|
\le |R_h(q+\sigma Z)-q|
\le \sigma|Z|+\frac h2.
\]
This remains true when the noisy rounded value lies outside the clipping interval.

\subsection{Gaussian concentration constants}

The map $F(z)=\sum_qc_q|z_q|$ is $D$-Lipschitz because
\[
|F(z)-F(z')|
\le\sum_qc_q|z_q-z'_q|
\le D\|z-z'\|_2.
\]
Its expectation is $\mu C$, where $\mu=\sqrt{2/\pi}<5/6$.  The standard Gaussian concentration inequality for Lipschitz functions gives
\[
\PP\{F\ge\mu C+t\}\le e^{-t^2/(2D^2)}.
\]
At $\Gamma=m/60$ and $\sigma_\star=2^{-11}$, and using only $h\le\sigma_\star$, the rational lower bound \eqref{eq:lambda-lower} follows.  The exact Gaussian mean would give the stronger per-clause tail separation $4.2819\ldots$ and exponent $0.2132\ldots m$; the paper uses the rational lower bounds $52m/15$ and $m/8$.

\section{Polynomial-bit sampling of the rounded Gaussian law}\label{app:sampling}

\subsection{One marginal}

Fix a base coefficient $q$, precision $B$, $h=2^{-B}$, and fixed rational $\sigma_\star$.  The support after rounding and clipping is
\[
G_B=\{-2,-2+h,\ldots,2-h,2\},
\qquad M=|G_B|=2^{B+2}+1.
\]
For an interior point $r$, the probability is \eqref{eq:atom}; the endpoint probabilities aggregate the corresponding half-infinite tails.  Every CDF argument is rational with $O(B)$ bits and lies in a fixed bounded interval depending only on $\sigma_\star$.  The standard normal CDF is polynomial-time computable as a real function \citep{ko1991real}; equivalently, one may combine polynomial-time evaluation of $e^{-x^2/2}$ on a fixed interval with quadrature or power-series approximation and elementary tail bounds.  Thus each atom can be approximated to $R$ binary digits in time $\poly(B,R)$.

Let the desired marginal total-variation error be $\varepsilon/P$.  Approximate each exact probability $p_i$ by a rational $a_i$ satisfying
\[
|a_i-p_i|\le\delta:=\frac{\varepsilon}{32PM},
\]
and replace negative approximations by zero.  The resulting vector has $\ell_1$ error at most $M\delta$, and its sum differs from one by at most the same amount.  Normalizing it increases the $\ell_1$ error by only a constant factor.  Next choose a dyadic denominator $2^R$ with
\[
2^R\ge\frac{16PM}{\varepsilon},
\]
round the first $M-1$ normalized probabilities downward to multiples of $2^{-R}$, and place the leftover mass on the last atom.  This changes the $\ell_1$ distance by at most $2M2^{-R}$.  With the displayed choices, the final dyadic marginal lies within $\varepsilon/P$ total variation of the exact marginal.  It can be sampled exactly using fair random bits and a cumulative table.  All stored integers and all arithmetic have polynomial bit length.

\subsection{Independent products}

Let $D_q$ and $\widehat D_q$ be the exact and approximate marginal laws.  A coordinatewise maximal coupling gives
\[
\TV\left(\bigotimes_{q=1}^P D_q,
          \bigotimes_{q=1}^P\widehat D_q\right)
\le\sum_{q=1}^P\TV(D_q,\widehat D_q)
\le\varepsilon.
\]
Any randomized verifier together with a timeout rule is a Markov kernel.  Total variation contracts under Markov kernels, so replacing the exact perturbation vector by the approximate sampler changes the probability of every final event by at most $\varepsilon$.

\subsection{Expected time and the hard-coded polynomial}

The smoothed polynomial bound for a fixed algorithm $A$ means that some fixed polynomial $Q$ upper-bounds its expectation.  A machine proving the complexity implication may hard-code $Q$ and run $A$ for $100Q$ steps.  This does not require discovering $Q$ from black-box access to $A$.  Markov's inequality then gives the $1/100$ timeout bound used in \cref{sec:bpp}.

\section{Experimental methodology}\label{app:experiments}

\subsection{MILP formulation}

For each hidden unit $y=\relu(z)$ with exact interval bounds $L\le z\le U$ on the box, stable units are represented by $y=z$ if $L\ge0$ and $y=0$ if $U\le0$.  For $L<0<U$, with binary phase variable $a$, we use
\[
y\ge0,
\qquad y\ge z,
\qquad y\le Ua,
\qquad y\le z-L(1-a).
\]
Because the network has one hidden layer, $L$ and $U$ are obtained exactly by selecting the appropriate endpoint of each input coefficient.  The models were solved through SciPy's MILP interface \citep{virtanen2020scipy}, backed by HiGHS; relevant dual-simplex implementation details are described by \citet{huangfu2018parallelizing}.  We required reported optimality and a relative MIP gap of $10^{-10}$.

\subsection{Arrangement enumeration}

\begin{lemma}[Three-dimensional arrangement vertex lemma]\label{lem:arrangement}
Let $g$ be a shallow ReLU network on a three-dimensional box.  A global maximum occurs at a box corner or at the intersection of three linearly independent planes chosen from the hidden activation planes and box facets.
\end{lemma}

\begin{proof}
The activation planes and box facets partition the box into finitely many compact polyhedral cells.  On each cell, every ReLU phase is fixed and $g$ is affine.  An affine function on a compact polytope attains a maximum at a vertex.  Every vertex in three dimensions is supported by three linearly independent active planes after including the planes defining any lower-dimensional cell and the box facets.
\end{proof}

For each 32-ReLU block, the solver enumerated all $\binom{38}{3}=8{,}436$ plane triples by vectorized Cramer's rule, discarded singular systems and points outside $[0,1]^3$, and evaluated the network at all remaining candidates.  Floating-point arithmetic is used only in this numerical study.  The independent MILP cross-check covered 44 perturbed networks; maximum and mean absolute discrepancies were $7.99\times10^{-15}$ and $3.63\times10^{-16}$.

\subsection{Product construction}

If blocks use disjoint inputs and their outputs are added, then
\[
\max_{x^{(1)},\ldots,x^{(q)}}\sum_{j=1}^qg_j(x^{(j)})
=\sum_{j=1}^q\max_{x^{(j)}}g_j(x^{(j)}).
\]
For the eight-clause satisfiable block, the base optimum is $8/3$; for the gap block it is $7/3$.  The midpoint threshold for $q$ homogeneous blocks is
\[
\tau_q=\frac{8q}{3}-\frac q6,
\]
leaving margin $q/6$ on both sides.  The shared output bias is perturbed once per product.  The bootstrap experiment used 300 exact block optima per family/noise cell and 50,000 products per product cell.

\subsection{Certificate experiment}

For one compact block, the vector of semantic parameters contains 112 slots.  Each realized perturbation was scored with the exact weighted sum in \eqref{eq:certificate}; the output-bias displacement was sampled separately.  At each of 37 noise levels, 12,000 single-block weighted sums were generated and then combined into 40,000 products for each $q$.  The 50\% transition was estimated by linear interpolation of certificate probability against $\log\sigma$.  Least-squares fits of $\log\sigma_{1/2}$ against $\log q$ gave slopes $-0.003496$ and $-1.013715$ for extensive and one-defect gaps.

\subsection{Figure typography}

All plots were generated with Matplotlib's PGF backend, $\texttt{text.usetex=True}$, and a LaTeX preamble containing \texttt{amsmath} and \texttt{amssymb}.  Legends were anchored above or outside their axes, so no legend or annotation obscures a data curve.  PDF and PGF versions are supplied.

\section{Proof-audit checklist}\label{app:audit}

The following items address common failure modes in smoothed-hardness arguments.

\begin{enumerate}[leftmargin=1.65em,itemsep=3pt]
\item \textbf{Smoothed adversarial, not random-instance, hardness.}  The formula and base network are worst-case.  Probability is only over independent coefficient perturbations.
\item \textbf{Fractional inputs.}  \Cref{thm:identity} proves a global rounding inequality on the entire box; no Boolean-input restriction is assumed.
\item \textbf{Exact correctness.}  The hypothetical verifier is exact on every rational input.  Approximation is used only to implement the outer reduction's sampling law in total variation.
\item \textbf{ReLU boundaries.}  The sensitivity proof uses only global 1-Lipschitzness and is valid at preactivation zero.
\item \textbf{No activation independence.}  No step treats ReLU phases or activation margins as independent random variables.
\item \textbf{Aggregate perturbation.}  The proof controls a weighted sum over all coefficients; it does not replace the aggregate by a single-coordinate anti-concentration claim.
\item \textbf{Clipping and rounding.}  Both are included explicitly in \eqref{eq:scalar-drift}.  Clipping is nonexpansive relative to every base coefficient in $[-1,1]$.
\item \textbf{Output perturbations.}  Hidden-to-output weights and the output bias are perturbed and counted.  The precision-guard output weight is also perturbed, although its first-order influence coefficient is zero because the base activation is zero.
\item \textbf{Strict witness.}  On the preservation event, the satisfiable witness remains strictly above the unchanged threshold.
\item \textbf{Fixed admissible noise.}  The guard forces $B\ge20$ at every size, so $2^{-B}<2^{-11}$ without asymptotic or finite-prefix assumptions.
\item \textbf{Expected versus bounded-error time.}  Markov truncation is applied to the joint expectation over the perturbation and verifier randomness.  Total variation is applied only after the timeout makes the final event bounded.
\item \textbf{Polynomial sampling.}  The rounded marginal support has polynomial cardinality because $B=O(\log s)$; Gaussian atom probabilities are approximated and normalized in polynomial bit time.
\end{enumerate}

The supplied \texttt{proof\_checks.py} performs additional redundant diagnostics: 64 random exact-identity checks (including repeated and complementary literals), 32 global sensitivity-certificate checks in which the true sup-norm difference is itself solved by MILP, and exact rational verification of every numerical constant.  The current audit reports maximum identity error $9.77\times10^{-15}$, positive sensitivity slack on every test, and status \texttt{PASS}.  These computations are diagnostics, not substitutes for the proofs.

\end{document}